\documentclass[10pt, a4paper]{amsart}

\usepackage{latexsym}
\usepackage[all]{xy}
\usepackage{color}
\newtheorem{teorema}{Theorem}[section]
\newtheorem{definicion}[teorema]{Definition}
\include{amslatex}
\newtheorem{proposicion}[teorema]{Proposition}

\newtheorem{corolario}[teorema]{Corollary}

\newtheorem{ejemplo}[teorema]{Example}

\newtheorem*{theorem-A}{Theorem A}
\newtheorem*{theorem-B}{Theorem B}
\newtheorem*{theorem-C}{Theorem C}
\newtheorem*{theorem-D}{Theorem [Akdar-Zadeh]}

\numberwithin{equation}{section}
\begin{document}

\begin{title}[Generalized Berwald spacetimes and the equivalence principle]{On singular generalized Berwald spacetimes and the equivalence principle}
\end{title}
\maketitle
\author{

\maketitle
\begin{center}
\author{Ricardo Gallego Torrom\'e\footnote{email: rigato39@gmail.com}}
\end{center}
\address{Departamento de Matem\'atica, Universidade Federal de S\~ao Carlos, Rodovia Washington Lu\'is, km 235, SP, Brazil}
\begin{abstract}
{\small The notion of singular generalized Finsler spacetime and singular generalized Berwald spacetime are introduced and their relevance for the description of classical gravity discussed. A  method to construct examples of such generalized Berwald spacetimes is sketched. The method is applied at two different levels of generality. First, a class of flat, singular generalized Berwald spacetimes is obtained.  Then in an attempt of further generalization, a class of non-flat generalized Berwald spacetimes is presented and the associated Einstein field equations are discussed. In this context, an argument in favour of a small value of the cosmological constant is given. The physical significance of the  models is briefly discussed in the last section.}
\end{abstract}

\section{Introduction}
Finsler spacetimes models are natural generalizations of Lorentzian metric theories of gravity. There are several reasons supporting such naturalness. Given a Finsler spacetime, a re-parametrisation invariant parameter that generalizes the notion of proper time in relativistic spacetimes and that also satisfies Einstein's clock hypothesis can be associated to each timelike curve. Therefore, they  provide a natural mathematical representation for the notion of ideal co-moving clock associated to a classical observer co-moving with the point particle \cite{Syngespecial1972}. From a more technical point of view, Finsler spacetimes are certainly of relevance, since fundamental techniques and notions from Riemannian and Lorentzian geometry can be generalized to the Finslerian case quite straightforwardly \cite{BaoChernShen2000}. In particular, in Finsler spacetimes there is a notion of geodesic as solution of the auto-parallel curves of a particular linear connection determined by the metric structure and there is a notion of curvature associated to such connection. There is also a  Lagrangian which is preserved along such auto-parallel curves.

 Although the richness of the Finsler category is an advantage for the many different applications in mathematics, physics and in other natural sciences and in engineering, such large variety of possible Finsler models originates a serious difficulty when the theory is applied to the description of the classical gravitational interaction. This is because for a generic four-dimensional manifold there are many more admissible Finsler spacetime structures than Lorentzian structures. For instance, one can consider small linear deformations of an original Lorentzian structure \cite{Randers}. Therefore, to have a natural, physically motivated criteria to identify the relevant Finsler spacetime structures is certainly of importance in the construction of  Finslerian models of gravity.

A possibility to solve this situation is to consider Berwald structures as the natural setting for Finslerian gravity. They  are  Finsler spacetimes whose associated linear Berwald connection is affine \cite{Bao2007}. Berwald spacetimes have been considered in the physics literature in connection with cosmological models \cite{Li}, in the discussion of higher order deviation equations \cite{Gallego Torrome Gratus} and also in connection with modified theories of relativity \cite{FusterPabst}. Berwald spacetimes are more similar to Lorentzian spacetimes than a generic Finsler spacetime in  the sense that smooth free falling coordinates systems, an essential requirement for the formulation of the equivalence principle in the Einstein or the strong form \cite{ThorneLeeLightman}, are possible to be constructed.

 Partially based on the above arguments, we postulate that in the category of Finsler spacetimes, it is natural that the following properties hold good:
 \begin{itemize}

 \item There is a natural notion for {\it timelike}, {\it lighlike} and {\it spacelike} curves,

 \item There is a definition of re-parametrization invariant natural time for {\it timelike curves},

 \item Einstein's clock hypothesis holds good,

 \item There is a linear connection determined by the metric structure and natural constraints,

 \item  There is a definition of geodesic as auto-parallel curves of a linear connection,

 \item There is a Lagrangian function which is constant along each parameterized auto-parallel curves of the connection,

 \item There is a well-defined theory of curvature associated to the connection,

 \item There are smooth local Fermi coordinate systems where the connection coefficients along a given geodesic are zero.

 \end{itemize}
No all these properties are required to be satisfied in a generic geometric spacetime theory. For instance, there are models where the geodesics are not auto-parallel curves of a linear connection (for instance, tele-parallel models of gravity). However,  the above properties are assumed to be satisfied in the spacetime models that in this work are considered, which are of Finsler type or direct generalizations. 

The above properties are satisfied for Berwald spacetimes.
However,  these properties are also satisfied in the larger class of {\it generalized Berwald spacetimes}, a notion which is closely related with the notion of deformed Finsler spaces \cite{Bucataru Miron, Vacaru Stavrinos} of Berwald type. In a Finsler spacetime the geometry is determined by a Finsler or Lagrangian function. Indeed, the {\it fundamental tensor} or {\it metric} is the {\it vertical Hessian} of the Lagrangian. In contrast, in a generalized Finsler spacetime the fundamental object is a covariant, symmetric non-degenerate $d$-tensor, namely, the fundamental tensor or metric. There is also associated a Lagrangian, but the metric is not the Hessian of the Lagrangian. On the other hand, a Berwald structure (of Euclidean or Lorentzian signature; Finslerian or Lagrangian; proper or generalized) is a structure characterized by the fact that a given connection is basically an affine connection. Generalized Finsler spacetimes of the Berwald type have the above  postulated properties and supported by this, we argue that the relevant generalized Finsler spacetimes for applications in gravity are within  the class of {\it generalized Berwald spacetimes}.  We will describe examples of generalizations of general Relativity that fall in such class that could have relevance in physics.

 The structure of this paper is the following. In {\it section 2}, a very short introduction to the notion of (singular) generalized Finsler spacetimes is presented. After this and motivated by their relation with the existence of {\it smooth free falling local coordinate systems} in the form of smooth Fermi coordinates, the notion of generalized Berwald spacetime is introduced. In {\it section 4},  a family of flat generalized  Berwald spacetimes is discussed, as  summarized in {\bf Theorem A}.  Then in {\it section 5} we outline a  method to extend the theory to obtain non-flat, generalized Berwald spacetimes. This leads us  in a straightforward way to {\bf Theorem B}. Furthermore, it is shown that the metrics considered in this paper satisfy the Einstein field equations in the case that there is no cosmological constant term.  If a non-zero cosmological constant term is included in the field equation, it is argued why the constant must be very small. In the last {\it section} the physical significance of our metrics is briefly discussed.

\section{Generalized Finsler spacetimes}
\subsection*{Notation}In this paper $M$ is a smooth four-dimensional spacetime manifold. $TM$ is its tangent bundle.  Greek letter indices stand for spacetime indices and run from $0$ to $3$. In the few cases that they will appear, Latin indices will run only from $1$ to $3$ and will correspond to spacelike directions in a particular $1+3$ decomposition product of the spacetime manifold. Repeated up and down indices are understood as summed from $0$ to $3$, if anything else is not stated. $\eta$ is the Minkowski metric in the special case when the spacetime  manifold $M$ is diffeomorphic to $\mathbb{R}^4$, while $h$ is a generic Lorentzian metric defined on $M$.  $^h\nabla$ is the Levi-Civita connection of $h$. We assume the standard definition of linear connection in a vector bundle $\pi_{\mathcal{E}}:\mathcal{E}\to M$  \cite{KobayashiNomizuI}, while the notion of affine connection that we adopt is that of a linear connection in the tangent bundle $\pi:TM\to M$. The definition of endomorphism curvatures follow also reference \cite{KobayashiNomizuI}. The main references that we follow for Finsler geometry are \cite{Bao2007} for the theory of linear Berwald connections and Berwald spaces and \cite{BaoChernShen2000} for general notions and results on Finsler geometry. The main reference that we follow for generalized Finsler spacetimes is \cite{Bucataru Miron}, from where the notion of $d$-tensor is taken.
Note that we have simplified the formulation of the curvature. The exact definition of curvature can be found in  \cite{Bao2007} (\cite{GallegoPiccioneVitorio:2012} for the analogous theory of the related Chern's connection) and requires the introduction of several canonical projections on the pull-back bundle $\pi^* TM\to\tilde{N}$. Thus the Berwald connection of a Berwald space is identified almost directly with the corresponding affine connection that it determines.

\subsection*{Finsler spacetimes}
In this work a weak version of Beem's definition  of Finsler spacetime \cite{Beem1970} to the category of generalized Finsler spaces  is adopted. This allows  us to consider generalized Finsler spacetimes with differential singularities of the metrics over certain sub-manifolds of $TM$. The domain of regularity of the metric is a sub-manifold $\tilde{N}\subset TM\setminus \{0\}$ such that the restriction to $\tilde{N}$ of the canonical projection $\tilde{\pi}:\tilde{N}\to M$ determines a  sub-bundle of $TM$ over $M$ and the fiber over $x\in M$ is $\tilde{N}_x:=\tilde{\pi}^{-1}(x)$. The notion of $d$-tensor is taken as in \cite{Bucataru Miron}.
\begin{definicion}
A singular generalized Finsler spacetime is a triplet $(M,\tilde{N},g)$ such that:
\begin{itemize}
\item $M$ is a four dimensional, Haussdorff, real, $C^{\infty}$-smooth  manifold.

\item $\tilde{\pi}:\tilde{N}:=\,TM\setminus (S\cup\{0\}) \to M$ is a fibre bundle such that the fiber $\tilde{N}_x$ is an open cone for each $x\in M$.

\item The {\it fundamental tensor} $g$ is a real, smooth, symmetric $d$-tensor defined on $\tilde{N}$ such that
\begin{itemize}
\item Its components $g_{\mu\nu}(x,\dot{x})$ are $0$-homogeneous in the local velocity coordinates $\dot{x}^\mu$,
\begin{align}
g_{\mu\nu}(x,\lambda\dot{x})=\,g_{\mu\nu}(x,\dot{x}),\quad \forall\, \lambda\in [0,\infty[.
\label{homogeneouscondition}
\end{align}
\item The matrix formed by the components $g_{\mu\nu}(x,\dot{x})$ is non-degenerate and with signature $(-,+,+,+)$ at each point $(x,\dot{x})\in\, \tilde{N}$.

\item For each point $x\in M$, the {\it null set}
\begin{align*}
^g{Con}_x:=\{(x,\dot{x})\in \tilde{N}_x\,|\,g_{\mu\nu}(x,\dot{x})\dot{x}^\mu\dot{x}^\nu=0\}
\end{align*}
 is a regular  open hypersurface of $\tilde{N}_x$ with only two connected components.
\end{itemize}
\end{itemize}
\label{Finslerspacetime}
\end{definicion}
The Lagrangian associated to the $d$-tensor of a tangent vector $\dot{x}\in\,T_xM$ is defined by the expression
\begin{align}
L(x,\dot{x})=\,g_{\mu\nu}(x,\dot{x})\dot{x}^\mu\dot{x}^\nu.
\label{definicion de norma}
\end{align}
Note that for a generic generalized Finsler spacetime the fundamental tensor $g_{\mu\nu}(x,\dot{x})$ is not the vertical Hessian of a given Finsler function, in particular for a generic generalized Finsler spacetime $g_{\mu\nu}\neq \frac{\partial^2 L}{\partial \dot{x}^\mu\dot{x}^\nu}$. On the other hand, if the equality $g_{\mu\nu}= \frac{\partial^2 L}{\partial \dot{x}^\mu\dot{x}^\nu}$ holds good, then the spacetime  determined by $g$ is a genuine Finsler spacetime. In this sense and apart from the causality and regularity conditions,  our {\it definition} is a particular case of generalized Lagrange spacetime and such that the homogeneity condition, characteristic of Finslerian structures, is imposed.

 {\it Definition} \ref{Finslerspacetime} allows for differentiable singularities of $g$ in domains of each slit tangent space $T_xM\setminus 0$. The region where the function is not regular is $S=\,\{S_x,\,x\in M\}$, where each $S_x$  contains the {\it singularity loci} of $g$ at $x\in M$ as well as the points $\dot{x}\in\,T_xM$ where $g_{\mu\nu}(x,\dot{x})$ are not defined real. The metrics that we will discuss in this paper contain differential singularities in the fundamental tensor $g$.

A vector $\dot{x}\in\, T_xM$ is either: 1. {\it Timelike} if $L(x,\dot{x})<0$, 2.  {\it Spacelike} if $L(x,\dot{x})>0$ or 3. {\it Lightlike} if $L(x,\dot{x})=0$.
 If there exists a smooth vector field $T\in \,\Gamma \tilde{N}$ such that at each point $x\in M$ the vector $T(x)$ is {timelike}, then there is defined a {\it time orientation}. In this case, a tangent vector $\dot{x}\in\,T_xM$ is future pointed if
 \begin{align*}
 g_{\mu\nu}(x,\dot{x})\,\dot{x}^\mu\,T^{\nu}(x)<0.
 \end{align*}
 Timelike curves  are smooth curves   such that the condition
\begin{align*}
L\left(x(t),\frac{dx(t)}{dt}\right)<0
\end{align*}
holds good along ${\bf x}:I\to M$.
Finslerian {\it lightlike} curves can also be defined. They are curves such that the condition
\begin{align*}
L\left(x(t),\frac{dx(t)}{dt}\right)=0
 \end{align*}
holds good along  ${\bf x}:I \to M$.

The Cartan tensor $C$ of the fundamental tensor $g$ is defined by its components $C_{\mu\nu\rho}$, which in terms of natural coordinates on $TM$ are given by the expression
\begin{align}
 C_{\mu\nu\rho}=\,\frac{1}{2}\frac{\partial}{\partial \dot{x}^\mu}\, g_{\nu\rho},\quad \mu,\nu,\rho =0,1,2,3.
\end{align}
By the zero $0$-homogeneity property of the fundamental tensor $g_{\nu \sigma}$ and by application of Euler's theorem of homogeneous functions, it follows the relation
\begin{align*}
\dot{x}^\mu\,C_{\mu\nu\rho}(x,\dot{x})=\,\frac{1}{2}\,\dot{x}^\mu\,\frac{\partial}{\partial \dot{x}^\mu}g_{\nu\rho}(x,\dot{x})=0.
\end{align*}
There are two further properties of the Cartan tensor that we shall mention. The first is that the structure determined by $g$ is Lorentzian iff the Cartan tensor is zero, by the same reason than in the Euclidean case \cite{BaoChernShen2000}. The second is that the structure $(M,\tilde{N},g)$ is a genuine Finsler spacetime, that is, the fundamental tensor $g$ is the vertical Hessian of a Lagrangian iff the Cartan tensor is totally symmetric \cite{Matsumoto, Bucataru Miron}.

Given a timelike curve ${\bf x}:I\to M$, a re-parametrisation invariant  parameter is defined by the expression
\begin{align}
\tau[{\bf x}](t,t_0):=\int^t_{t_0}\,\sqrt{-L\left(x(s),\frac{dx(s)}{ds}\right)}\,ds.
\label{propertime}
\end{align}
We call $\tau[{\bf x}]$ the {\it natural time parameter} of the Finsler metric $g$ or proper time parameter of $g$.
Note that the parameter $s$ has not been fixed to be a particular parameter, neither it is necessary to specify a dynamical law for the curves ${\bf x}:I\to M$.

We can compare the definition of natural time parameter \eqref{propertime} and the standard definition of proper time in general relativity. Given a relativistic spacetime $(M,h)$, the proper time along the timelike curve ${\bf x}:I\to M$  is given by the expression
\begin{align}
\tau_h[{\bf x}](t,t_0):=\,\int^t_{t_0}\,\sqrt{-h_{00}(x(t))\,\dot{x}^0\dot{x}^0}\,dt
\label{propertimeforh}
\end{align}
in a co-moving reference frame where $\dot{x}^i=0,\,i=1,2,3$. This parameter $\tau_h(t)$ is the interval of time that  an ideal clock co-moving with a particle along the world line ${\bf x}:I \to M$ measures. It is direct that the expression \eqref{propertimeforh} can be re-written in invariant form under local coordinate transformations,
\begin{align}
\tau_h[{\bf x}](t,t_0):=\,\int^t_{t_0}\,\sqrt{-h(\dot{x},\dot{x})}\,dt.
\label{relativisticpropertime}
\end{align}
Therefore, the definition of natural time parameter \eqref{propertime} of a generalized Finsler spacetime contains the relativistic notion of proper time \eqref{relativisticpropertime} as a particular case when $g$ corresponds to the Lorentzian metric $h$.

A fundamental difference between genuine Finsler spacetimes and generalized Finsler spacetimes comes as follows. Let us consider the Euler-Lagrange equations
\begin{align*}
\frac{d}{d\tau}\frac{\partial L}{\partial \dot{x}^\mu}-\,\frac{\partial L}{\partial x^\mu}=0,\quad \mu=0,1,2,3,
\end{align*}
for the energy functional $E$,
\begin{align}
E[{\bf x}]=\,\int_{I}\,d\tau\,L(x,\dot{x}).
\label{Energy functional}
\end{align}
For the Lagrangian \eqref{definicion de norma}, the Euler-Lagrange  are the differential equations
\begin{align*}
\frac{d}{dt}\left(\frac{\partial}{\partial \dot{x}^\mu}\,g_{\rho\nu}(x,\dot{x})\right)\,\dot{x}^\rho\dot{x}^\nu+\,2\,g_{\mu\nu}(x,\dot{x})\,\ddot{x}^\nu\,-\frac{\partial}{\partial x^\mu}\,g_{\rho\nu}(x,\dot{x})\dot{x}^\rho\dot{x}^\nu=0.
\end{align*}
If one introduces the Cartan tensor, these equations are re-written as
\begin{align}
2\,\dot{x}^\rho\dot{x}^\nu\,\frac{d}{dt}A_{\mu\rho\nu}(x,\dot{x})\,+2\,g_{\mu\nu}(x,\dot{x})\,\ddot{x}^\nu\,-\frac{\partial}{\partial x^\mu}\,g_{\rho\nu}(x,\dot{x})\dot{x}^\rho\dot{x}^\nu=0,\quad \mu,\nu,\rho =0,1,2,3.
\label{Euler Lagrange of E}
\end{align}
This equations can be compared with the usual expression obtained for a Finsler spacetime, that due to the totally symmetry of the Cartan tensor and the $0$-homogeneity of the fundamental tensor $g$, is of the form
\begin{align}
2\,g_{\mu\nu}(x,\dot{x})\,\ddot{x}^\nu-\frac{\partial}{\partial x^\mu}\,g_{\rho\nu}(x,\dot{x})\dot{x}^\rho\dot{x}^\nu=0,\quad \mu,\nu,\rho =0,1,2,3.
\label{Euler Lagrange of a Finsler spacetime}
\end{align}
Using the non-degeneracy and symmetry of $g$, the equations \eqref{Euler Lagrange of a Finsler spacetime}  are re-cast as
\begin{align}
\ddot{x}^\mu+\,\gamma^\mu\,_{\nu\rho}(x,\dot{x})\dot{x}^\nu\,\dot{x}^\rho=0,\quad \mu,\nu,\rho =0,1,2,3,
\label{geodesic equations}
\end{align}
where $\gamma^\mu\,_{\nu\rho}(x,\dot{x})$ are the Christoffel symbols of $g$,
 \begin{align}
\gamma^\mu\,_{\nu\rho}[g]:=\,\frac{1}{2}\,g^{\mu\sigma}\,\left(\frac{\partial g_{\sigma\rho}}{\partial x^\nu}\,+\frac{\partial g_{\nu\sigma}}{\partial x^\rho}\,-\frac{\partial g_{\nu\rho}}{\partial x^\sigma}\right).
 \label{Christoffel symbols of g}
 \end{align}
We observe that an extra term appears in the Euler-Lagrange equations \eqref{Euler Lagrange of E} of a generalized Finsler spacetime respect to the usual expression  of the Euler-Lagrange equations \eqref{Euler Lagrange of a Finsler spacetime} of a genuine Finsler spacetime, which are equivalent to the  equations \eqref{geodesic equations}. The appearance of this extra term is the main difference between genuine Finsler spacetimes and generalized Finsler spacetimes.

 In order to formulate the theory of generalized Finsler spacetimes in a general covariant way it is necessary a connection. In our case, we choose to work with the {\it linear Berwald connection} $^b\nabla$, which is a  linear, {\it symmetric} connection on the pull-back bundle $\pi^* TM$ over $\tilde{N}$. The details of the construction are analogous to the construction in the Euclidean case \cite{Bao2007}. The linear connection $ ^b\nabla$ is characterized by: 1. Being a symmetric and 2. The auto-parallel curves of the connection, that is, the curves on $M$ such that
\begin{align}
^b\nabla_{\dot{x}}\,\dot{x}=0.
\label{geodesic equation 2}
\end{align}
coincide with the solutions \eqref{geodesic equations}.
We call the solutions of this differential equation the {\it parameterized geodesics} of $g$. One can shown by a simple argument in local coordinates that if the non-zero connection coefficients are given by \eqref{Christoffel symbols of g}, then the equation \eqref{geodesic equation 2} is equivalent to the equation \eqref{geodesic equations}.

Note that in the class of generalized Finsler spacetimes there is a radical distinction between the auto-parallel condition \eqref{geodesic equation 2} and the Euler-Lagrange equation \eqref{Euler Lagrange of E}. The motivation for our choice of equation \eqref{geodesic equation 2}  as the geodesic equation is partially discussed in {\it Section 3}, where we will show that our notion is compatible with the existence of smooth local Fermi coordinates and normal coordinates on $M$  in a special class of generalized Finsler spacetimes.

It is direct from the definition of this connection that
\begin{align}
^b\nabla g=0.
\label{covariant conservation of g}
\end{align}
Moreover, along any solution of the geodesic equation  \eqref{geodesic equation 2}
\begin{align*}
^b\nabla g(\dot{x},\dot{x})=\,^b\nabla g(\dot{x},\dot{x})\,+2\,g(\,^b\nabla_{\dot{x}}\dot{x},\dot{x})=0,
\end{align*}
which implies that the Lagrangian $L(x,\dot{x})$ is preserved,
\begin{align}
\frac{d}{dt} \,L(x,\dot{x})=\,\frac{d}{dt}\,(g(\dot{x},\dot{x}))=0.
\label{covariant conservation of L}
\end{align}

For a generic generalized Finsler spacetime the connection coefficients of $^b\nabla$ depend on the points $(x,\dot{x})\in \tilde{N}$ except in the very special case of Finsler spacetimes of Berwald type, to be defined in the following {\it section}, in which case they only depend on the spacetime point $x\in M$.

The last technical notion that we need to introduce to complete our geometric tool-kit  is the curvature of the connection $^b\nabla$. The Riemannian curvature endomorphisms of the linear Berwald connection are given by the expression
\begin{align}
R_g(X,Y)Z=\,^b\nabla_X\,^b\nabla_Y Z-\,^b\nabla_X\,^b\nabla_Y Z-\,^b\nabla_{[X,Y]}Z,\quad X,Y,Z\in \,\Gamma \tilde{N}.
\label{curvature endomorphism}
\end{align}

\section{Free falling local coordinate systems, gravity and generalized Berwald spacetimes}
 The existence of smooth free falling local coordinate systems is not a necessary formal requirement for the formulation of the universality of free fall or weak equivalence principle (WEP), that only requires the universality of the geodesic equation, if the geodesics are  associated to the world lines of free falling test particles. This point is illustrated with the following example.
 \begin{ejemplo}
Let us consider a {\it toy universe} where the {\it geodesics} are the solutions of the Lorentz force equation and the physical systems under consideration are constrained  to have the same quotient $q/m$ between charge and mass.  In such {\it toy universe}, even if the associated linear Berwald connection is not affine \cite{Ricardo2014}, there is universality of the connection coefficients and the associated proper time functional does not depend on the particular specie of particle characterized by the pair $(q,m)$. The conjunction of these two facts can be interpreted as { universality of free fall in this toy universe}, that is, WEP holds in such universe. However, since the Lorentz force equation is geometrized in the form of a geodesic equation of a Randers spacetime \cite{Randers}, this example shows that in the  finslerian category the WEP does not imply the Einstein Equivalence Principle \cite{ThorneLeeLightman}. Indeed, for a generic Randers spacetime it is not possible to formulate Einstein's version of the equivalence principle, since for a generic Randers spacetime it is impossible to define smooth free falling local coordinate systems where all the gravitational effects, due for instance to the connection, are eliminated in small enough region. For instance, gravitational effects appear in the geodesic equation in the piece proportional to the  connection coefficients and for a general Randers spacetime, cannot be eliminated along arbitrary geodesics \cite{BaoChernShen2000}.
 \end{ejemplo}

The existence of smooth local free fall coordinate systems is required for the formulation of the Einstein equivalence principle (EEP) and other stronger versions of the equivalence principle, but it is also essential in the definition of gravity as an interaction that can be eliminated in the description of test systems, for sufficient small macroscopic domains.
  It is in smooth free falling local coordinate systems that very specific {\it dynamical effects}, namely, effects that do not depend on the specific characteristics of the test body, can be eliminated in a small enough region of spacetime in an uniformly wa, independently of the state of motion of the test particles or test systems. By elimination of dynamical effects we mean that test particles will move as free particles in the sense that they world lines are represented by straight lines in such free fall local coordinate systems.  Hence any deviation from free motion will indicate the existence of another field different than the one acting in an universal way on test particles. Since the free falling local coordinate systems are smooth,  absolute tensor calculus makes not necessary to specify experimentally such coordinate systems in the problem of evaluation and contrast with experience of scalar observable functions.  Here the condition of smoothness is essential for the mathematical consistency of the tensorial equations under local coordinate transformations. At least $\mathcal{C}^2$-smoothness regularity is need for the theory, since the covariant dynamical equations are constructed using a connection. This scheme of things is our definition of gravity as a geometric phenomena and also provides a methodology of how to {\it measure} the action of gravity.

    In order that the above characterization of gravity could work in practice, smooth free falling local reference frames and the definition of the associated local coordinate systems should not depend on the relative state of motion of the test particle. To show why this must be the case, let us consider the opposite possibility  and let us assume that  free fall local coordinate systems where gravitational effects are eliminated depend on the state of motion of the point test particle with world line ${\bf x}:I\to M$. Let us consider another test particle with world line $\tilde{\bf x}:I\to M$ in relative motion respect to the particle with world line ${\bf x}:I \to M$. In the free fall coordinate system determined by ${\bf x}$ the particle following $\tilde{\bf x}$ will not be in free motion, because in such coordinate system the gravitational field has not been eliminated in its action along the world line $\tilde{\bf x}:I \to M$ but only in its action along the world line ${\bf x}:I\to M$. That is, there will be detectable effects associated to gravity at arbitrary small distances in the coordinate system free falling coordinate systems co-moving with ${\bf x}$. This implies the appearance of a logical ambiguity in the interpretation of the deviation of free motion of the particle with world line $\tilde{\bf x}$  between: 1. The possibility that gravity acts along $\tilde{\bf x}:I\to M$ in a way that depends on its specific characteristics  and the particular state of motion respect to another particle with world line ${\bf x}:I\to M$ and 2. The possibility  that an additional field is acting on the particle $\tilde{\bf x}$. This situation is in contradiction with the characterization of gravity discussed above.

Based on the above argument, in this paper we adopt the assumption that the differential structure of $M$ admits smooth free falling local coordinate systems. A natural way to implement mathematically that notion for a given generalized Finsler spacetime is to identify these coordinate systems  with {\it Fermi coordinate systems}.  Fermi coordinates are characterized by the fact that the Christoffel symbols of the fundamental tensor $g$  are zero along a given auto-parallel curve ${\bf x}:I\to M$. Therefore, we consider the following
 \begin{definicion}
 Let $(M,\tilde{N},g)$ be a generalized Finsler spacetime and ${\bf x}:I\to M$ be an auto-parallel curve of the linear Berwald connection $^b\nabla$ of the fundamental tensor $g$. A  free falling local coordinate system along ${\bf x}:I\to M$ is a Fermi coordinate system of $^b\nabla$.
 \label{freefallcoordinatesystem}
 \end{definicion}
 By local coordinates we mean coordinates on $M$. Since $^b\nabla$ is a connection defined action on sections over $\tilde{N}$ and not on $M$, the above definition seems at least ambiguous, except if we consider natural coordinates on $\tilde{N}$, induced from local coordinates on $M$. In this case, a free falling local coordinate systems is a coordinate system on $M$ that induces a natural coordinate system \cite{BaoChernShen2000} on $\tilde{N}$.

  The existence of local smooth Fermi coordinate systems imposes restrictions on the atlas structure of $M$ to make the theory compatible with the existence of the $d$-tensor $g$ defined on $\tilde{N}$. In Fermi coordinate systems, by definition, the Christoffel symbols $\gamma^\mu\,_{\nu\rho}$ calculated with $g$ vanish along the curve ${\bf x}:I\to M$. Then it can be shown by direct computation that, admitting $\mathcal{C}^2$-smoothness in the transition functions associated to the local coordinate transformations, in any other coordinate system the associated Christoffel symbols $\tilde{\gamma}^\mu\,_{\nu\rho}$ of  $g$  are defined over the spacetime manifold $M$. Therefore, the Christoffel symbols determine an affine, symmetric connection on $M$. Conversely, the construction of smooth local Fermi coordinates can be done along auto-parallel curves of any affine, symmetric connection (see for instance the argument given in \cite{Perlick2007}). It is this property, together with the covariant preservation of the Lagrangian, namely, equation \eqref{covariant conservation of L}, what justifies the adoption of the notion of geodesic as auto-parallel curves of an affine, symmetric Berwald connection $^b\nabla$.

  The existence of such type of connection for a given Finsler spacetime is one of the characterizations of Berwald spaces \cite{Bao2007}, that we adopt here as our definition of generalized Berwald spacetime,
 \begin{definicion}
 $(M,\tilde{N},g)$ is a generalized Berwald spacetime iff corresponding linear Berwald connection is the pull-back connection on $\pi^*TM$ of an affine connection on $M$.
  \end{definicion}
Our definition requires the notion of pull-back connection of an affine connection. Details of such constructions with applications to Berwald spaces of Euclidean signature can be found in \cite{Gallego Torrome Etayo}, but basically are equivalent to the characterization of $^b\nabla$ by the fact that it is symmetric and the connection coefficients live on $M$. For Lorentzian signature, the notion of pull-back connection remains the same.
  For positive definite Finsler metrics this characterization of Berwald space is well known (see \cite{Bao2007} or \cite{BaoChernShen2000}, Chap. 10 for the analogous characterization using the Chern connection).
  The translation of several other notions and results from the theory of Berwald spaces with Euclidean signature to the theory of generalized Berwald spacetimes of Lorentzian signature requires no major changes and will be used without proof. We refer to the interested reader to  \cite{Bao2007,BaoChernShen2000, Bucataru Miron, Matsumoto} for details.

  Because generalized Finslerian models also accommodate Einstein's clock hypothesis, the category of generalized Berwald spacetimes is naturally linked with a geometric description of gravitational phenomena. We develop the fundamental elements of a theory of generalized Berwald spacetimes and its application as gravitational models in the following {\it sections} of this paper. However, our theory is not exhaustive and there are examples falling in the category of Berwald spacetimes but laying  outside the applicability of the theory described below, as the ones investigated in \cite{GibbonsGomisPope2007, Li, FusterPabst}.

 \section{Flat generalized Berwald spacetimes}
Let us assume that the fundamental tensor $g$ defining the generalized Finsler spacetime is
\begin{align}
\tilde{g}_{\nu\rho}(\tilde{x},\dot{\tilde{x}})=\,\tilde{\eta}_{\nu\rho}(\tilde{x})\,+\tilde{\epsilon}_{\nu\rho}(\tilde{x},\dot{\tilde{x}})
\label{metricg}
\end{align}
with an  ansatz for the deformation factor $\phi$ of the form
\begin{align}
\epsilon_{\alpha\nu}(\dot{x})=\,\eta_{\alpha\nu}\,\phi(\dot{x}),
\label{antsatzforepsilon}
\end{align}
where $\phi:\tilde{N}\cup\{0\}\to \mathbb{R}^+$ is a positive, real scalar function. The factor $\phi$ only depends on the velocity coordinates $\dot{x}^\mu$, a property which is independent of the natural  coordinate systems on $\tilde{N}\cup\{0\}$ induced from local coordinate systems on $TM$.
 In order to ensure the property of $0$-homogeneity on velocity coordinates of the fundamental tensor $g$, it is necessary for the scalar factor $\phi$ to be a $0$-homogeneous function on the velocity coordinates $\dot{x}^\mu$.
 Hence the scalar field $\phi$ is assumed to be of the form
\begin{align}
\phi\left(\theta,l^2\,[\mathcal{R}_\eta],
\frac{\eta(\dot{x},\dot{x})}{\eta^2(\dot{x},\mathcal{A})},\frac{\eta(\dot{x},\dot{x})}{\eta^2(\dot{x},\mathcal{B})},...\right)
\label{verygeneralformoftheansatz}
\end{align}
 an analytic function on its arguments.
$\mathcal{A},\mathcal{B},...$ are timelike vector fields respect to the metric $\eta$. $[\mathcal{R}_\eta]$ is a short way to denote any possible dependence on curvature scalars formed from $\eta$ (ex.  Ricci scalar). Although for the Minkowski metric $\eta$ the curvatures in $[\mathcal{R}_\eta]$ are identically zero, we still keep this dependence, having in mind  the generalization to non-flat case that we shall consider in the next {\it section}. The parameter $\l$ has physical dimension of length and does not change under re-parameterizations of the time parameter, while the expression $l^2\,[\mathcal{R}_\eta]$ is homogeneous of degree zero in velocity coordinates. $\theta$ denotes all the dimensionless  scalar parameters on which the fundamental tensor $g$  depends on. The dots in the formal dependence of $\phi$ stand for additional $0$-homogeneous arguments. The time parameter $t$ is not fixed to be a particular parameter (ex. proper time of $\eta$ or the proper parameter of the candidate to fundamental tensor $g$), neither it is required that the curve ${\bf x}:I\to M$ satisfies an specific dynamical equation.

By the ansatz \eqref{antsatzforepsilon}, the fundamental tensor $g$ is of the form
\begin{align}
g_{\mu\nu}(\dot{x})=\,\left(1+\,\phi\left(\theta,l^2\,[\mathcal{R}_\eta],
\frac{\eta(\dot{x},\dot{x})}{\eta^2(\dot{x},\mathcal{A})},\frac{\eta(\dot{x},\dot{x})}{\eta^2(\dot{x},\mathcal{B})},...\right)\right)\,\eta_{\mu\nu}.
\label{metricberwaldphi}
\end{align}
By inspection of the corresponding Christoffel symbols we can see that the $d$-tensor \eqref{metricberwaldphi} is a generalized Berwald spacetime, since the connection coefficients of the Berwald connection do not depend on the $\dot{x}$-coordinates (remember that in our ansatz for the flat case, $\phi$ does not depend on $x$-coordinates). Indeed, one  can easily show that, because the form of the ansatz \eqref{metricberwaldphi}, in the local coordinate system where the metric $\eta$ is $diag(-1,1,1,1)$ the Christoffel symbols of $g_{\mu\nu}$ are also zero,
\begin{align*}
\gamma^\mu\,_{\nu\rho}[g]=\,\,\frac{1}{2}\,g^{\mu\sigma}\left(\partial_\nu g_{\sigma\rho}+\,\partial_\rho g_{\nu\sigma}-\partial_\sigma g_{\nu\rho}\right)=\,\frac{1}{2}\,\eta^{\mu\sigma}\left(\partial_\nu \eta_{\sigma\rho}+\,\partial_\rho \eta_{\nu\sigma}-\partial_\sigma \eta_{\nu\rho}\right)=0.
\end{align*}
Hence in any other coordinate system the  connection coefficients depend only on the coordinates of $x\in M$. Indeed, if we perform a change of coordinates on $M$, the Christoffel symbols $\tilde{\gamma}^\mu\,_{\nu\rho}[g]$ of $g$ are given in terms of the new components $\tilde{\eta}^{\mu\sigma}$ by the relation
\begin{align}
\tilde{\gamma}^\mu\,_{\nu\rho}[g]=\,
\tilde{\gamma}^\mu\,_{\nu\rho}[\eta]=\,\frac{1}{2}\,\eta^{\mu\sigma}\left(\partial_\nu \eta_{\sigma\rho}+\,\partial_\rho \eta_{\nu\sigma}-\partial_\sigma \eta_{\nu\rho}\right).
\label{equationforgammas}
\end{align}

Let us consider the regularity properties of the fundamental tensor \eqref{metricberwaldphi}.
In order to ensure the non-degeneracy  of the fundamental tensor $g$ on $\tilde{N}$ it is necessary to impose the condition
\begin{align*}
1+\,\phi\left(\theta,l^2\,[\mathcal{R}_\eta],
\frac{\eta(\dot{x},\dot{x})}{\eta^2(\dot{x},\mathcal{A})},\frac{\eta(\dot{x},\dot{x})}{\eta^2(\dot{x},\mathcal{B})},...\right)\neq 0\quad \forall \, (x,\dot{x})\in \,\tilde{N}.
\end{align*}
To avoid signature changes in $g_{\mu\nu}(x,\dot{x})$, for instance, the possibility of transitions of the form
\begin{align*}
diag(-1,1,1,1)\to \,diag(1,-1,-1,-1)
\end{align*}
 in the diagonal forms of the metrics $g_x$, it is necessary and sufficient to impose the stronger condition
\begin{align}
1+\,\phi\left(\theta,l^2\,[\mathcal{R}_\eta],
\frac{\eta(\dot{x},\dot{x})}{\eta^2(\dot{x},\mathcal{A})},\frac{\eta(\dot{x},\dot{x})}{\eta^2(\dot{x},\mathcal{B})},...\right) >0.
\label{boundforphi}
\end{align}

We make the assumption that the light cones of $\eta$  are embedded in the domain $\tilde{N}$.
If this assumption is adopted, then it is natural to further impose  that the null sets of $g$ coincide with the light cones of $\eta$. Otherwise,  there will be double light cone structures, in contradiction to our {\it definition} of Finsler spacetime \ref{Finslerspacetime} and also a situation which is unnecessary for our goals in this paper.

It is instructive to determine the Cartan tensor of \eqref{metricg} when the ansatz for $\epsilon(x,\dot{x})$ is given by the equation \eqref{antsatzforepsilon}. In this case we have
\begin{align*}
C_{\mu\nu\rho}[g]=\,\eta_{\nu\rho}\frac{\partial \phi_{\theta l\mathcal{A}\mathcal{B}...}}{\partial \dot{x}^\mu}.
\end{align*}
If we denote by $\chi=\eta(\dot{x},\dot{x})$ and by $\Theta_{\mathcal{A}}=\eta(\dot{x},\mathcal{A})$, $\Theta_{\mathcal{B}}=\eta(\dot{x},\mathcal{B})$, etc..., then the expression for the Cartan tensor components is of the form
\begin{align}
C_{\mu\nu\rho}[g] = \,\left[2\,\dot{x}^\sigma\,\frac{\partial \phi_{\theta l\mathcal{A}\mathcal{B}...}}{\partial \chi}\,+\frac{\partial \phi_{\theta l\mathcal{A}\mathcal{B}...}}{\partial \Theta_\mathcal{A}}\,\mathcal{A}^\sigma+\,\frac{\partial \phi_{\theta l\mathcal{A}\mathcal{B}...}}{\partial \Theta_\mathcal{B}}\,\mathcal{B}^\sigma\,+...\right]\,
\,\eta_{\nu\rho}\,\eta_{\mu\sigma},
\label{CartantensorofLphi}
\end{align}
which is in general non-zero.

The above arguments provide a proof of the following
 \begin{theorem-A}
 Let $(M,\tilde{N},g)$ be a generalized Finsler spacetime such that its fundamental tensor is  given by \eqref{metricberwaldphi}  and such that
 \begin{align*}
 \left(\frac{\partial \phi}{\partial \chi}, \frac{\partial \phi}{\partial \mathcal{A}},\,\frac{\partial \phi}{\partial \mathcal{B}},...\right)\,\neq \left(0,0,0,...\right).
 \end{align*}
 Then the Cartan tensor \eqref{CartantensorofLphi} is non-zero and $(M,\tilde{N},g)$ is a non-Lorentzian singular generalized Berwald spacetime.
  \label{Theorem A}
  \end{theorem-A}

\subsection*{Some properties of the singular generalized Berwald spacetime \eqref{metricberwaldphi}}

 The relation \eqref{equationforgammas} implies that the metric structures $\eta$ and $g$ are {\it un-parameterized geodesic equivalent} and that they define the same affine geometries, since the Levi-Civita connection of $\eta$ and the linear Berwald connection of $g$ are equivalent connections. Therefore, in order to compare the generalized Finsler spacetime metric $g$ given by \eqref{metricberwaldphi} with  the Minkowski spacetime it is necessary to investigate  the corresponding chronometric properties.  In the present case such properties do not depend on the curvatures, since the metric $g$ is flat, but they exhibit intrinsic finslerian properties due to the dependence of the metric on the velocity tangent field $\dot{x}:I \to TM$. In particular, the metric $g$ can be non-reversible.

 Let us denote the null set or light cone of $g$ at $x$ by
 \begin{align*}
 ^g{Con}_x:=\{(x,\dot{x})\in T_x M\,s.t.\,L_g(x,\dot{x})=0\}
  \end{align*}
  and the light cone of $\eta$ at $x$ by
  \begin{align*}
  ^\eta{Con}_x=\{(x,\dot{x})\in T_x M\,s.t.\,\eta(\dot{x},\dot{x})=0\}.
  \end{align*}
   Then we have as a consequence  of \eqref{boundforphi} the following
\begin{proposicion}
 For the metric \eqref{metricberwaldphi}, the relation
 \begin{align*}
  ^g{Con}_x=\,^\eta{Con}_x
  \end{align*}
  holds good for each $x\in M$.
\end{proposicion}

  Since \eqref{metricberwaldphi} corresponds to a generalized Berwald spacetime, the curvature endomorphisms \eqref{curvature endomorphism} are the only one that can be non-zero \cite{Bao2007}. However, as direct consequence of \eqref{equationforgammas} for the Minkowski metric $\eta$, the Riemannian curvature endomorphisms of the Berwald connection are also zero,
\begin{align}
R_g(X,Y)Z=\,R_\eta(X,Y)Z=0.
\end{align}
Hence the generalized Berwald metric \eqref{metricberwaldphi} is {\it Finsler flat}, that is, all its curvatures vanish, but it is in general different from the Minkowski metric $\eta$. This is a particular example of the Lorentzian version of a result of V. Matveev on geodesic rigidity \cite{Matveev2008}.

Let us remark that if the factor $\phi$ depends on the peculiarities of each individual  test particle, for instance, on the specific mass or charge, the WEP does not hold, although there are defined free falling local coordinates systems. This argument completes our previous claim that the conditions for WEP to hold and the existence of local free falling coordinate systems are logically independent.

\section{Non-flat generalized Berwald spacetimes}
 We can extend the methods and results of the previous {\it section} as follows. Let us consider a non-flat Lorentzian metric $h$.  Given a curve ${\bf x}:I\to M$, that for all effects will be either timelike or lightlike geodesic\footnote{The condition that ${\bf x}:I\to M$ is a geodesic can be eliminated for all practical purposes that follows in this section, since we can consider normal coordinates centered at $x$, that by Whitehead theorem exists for an affine connection \cite{Whitehead1932}.} respect to $h$, there are defined Fermi coordinates for $h$  along {\bf x}. In such Fermi coordinate system the Christoffel symbols of $h$ vanish, $^h\gamma^\mu\,_{\nu\rho}(x(t))=0$ and the metric takes the diagonal form
  \begin{align*}
  h(x(t))=\,diag(-1,1,1,1).
   \end{align*}
 The generalization from flat to non-flat spacetimes of the metric \eqref{metricberwaldphi}  is based on the  substitution of $\eta$ by $h$ in the ansatz for the fundamental tensor,
\begin{align}
g_{\mu\nu}(x,\dot{x})=\,\left(1+\,\phi(\dot{x})\right)\,h_{\mu\nu}(x).
\label{generalBerwaldtype}
\end{align}
Here $\phi$ is given as in \eqref{verygeneralformoftheansatz} but where $\eta$ is substituted by $h$.
   \begin{proposicion}
   If $(M,\tilde{N},g)$ is of the form \eqref{generalBerwaldtype} such that the $0$-homogeneous function $\phi:\tilde{N}\to \mathbb{R}^+$ is constant on $M$,
   \begin{align}
  \partial_\rho \phi(x,\dot{x})=0,\,\quad \,\forall\,\dot{x}\in\,\tilde{N}_x,
   \label{differential constraint on phi}
   \end{align}
   then  the linear Berwald connection of $g$ determines directly an affine connection in $M$.
   \end{proposicion}
   \begin{proof}
   Let us consider a fixed point $x\in \,M$. For a local Fermi coordinate with starting point at $x$ along the world line ${\bf x}:I\to M$, one has that $\gamma^\mu\,_{\nu\rho}[h]=0$. Hence it also holds that at $x\in \,M$
   \begin{align*}
   \partial_\rho \,h_{\mu\nu}=0,\quad \mu,\nu,\rho=0,1,2,3.
   \end{align*}
   Since $\partial_\rho \phi=0$, then one has  in the local Fermi coordinates of $h$
   \begin{align*}
   \partial_\rho \,g_{\mu\nu}(x,\dot{x})=0,\,\quad \mu,\nu,\rho=0,1,2,3
   \end{align*}
   along the geodesic ${\bf x}:I\to M$.
   Therefore, the Christoffel symbols of \eqref{generalBerwaldtype} are zero in such coordinate system and they can only depend on the coordinates of the point $x$ in any other coordinate system.
   \end{proof}
   \begin{corolario}
If $(M,\tilde{N},g)$ is a generalized Finsler spacetime such that the fundamental tensor is of the form \eqref{generalBerwaldtype} and the $0$-homogeneous function $\phi:\tilde{N}\to \mathbb{R}^+$ is such that the condition \eqref{differential constraint on phi} holds, then the connection coefficients of the Levi-Civita connection of $h$ coincide with the connection coefficients of the linear Berwald connection of $g$.
   \end{corolario}
   \begin{proof}
   Note that if in Fermi coordinates $\partial_\rho \,h_{\mu\nu}=0$, then $\partial_\rho\,g_{\mu\nu}=0$ also holds. It follows that since this holds for any $\rho, \mu,\nu=0,1,2,3$, it also holds  in Fermi coordinates the relation
   \begin{align*}
   h^{\sigma\rho}\,\partial_\rho \,h_{\mu\nu}=\,g^{\sigma\rho}\,\partial_\rho \,g_{\mu\nu}=0,\quad \sigma,\mu,\nu=0,1,2,3,
   \end{align*}
   from which it follows that the equality of the derivatives must hold in any local coordinate system on $M$ and as consequence $\gamma^\sigma_{\mu\nu}[h]=\,\gamma^\sigma_{\mu\nu}[g]$ in any local coordinate system.
   \end{proof}
Then we have proved the following
  \begin{theorem-B}
If for the generalized Finsler spacetime $(M,\tilde{N},g)$, the fundamental tensor is of the form \eqref{generalBerwaldtype} and for the $0$-homogeneous function $\phi:\tilde{N}\to \mathbb{R}^+$ the condition \eqref{differential constraint on phi} holds good, then $(M,\tilde{N},g)$  is a (singular) generalized Berwald spacetime.
  \label{Theorem B}
\end{theorem-B}

\begin{corolario}
With the same hypothesis than in {\bf Theorem B},
for the metric \eqref{generalBerwaldtype}, it holds that  $^gCon_x=\,^hCon_x$.
\end{corolario}
\begin{ejemplo}
 As an example, let us consider $(M,h)$ to be a Robertson-Walker spacetime,
 \begin{align}
 ds^2=\,-dt^2+\,a^2(t)\left(\frac{d{r}^2}{1-\epsilon \,r^2}+r^2(\sin^2\theta d{\phi}^2+d{\theta}^2)\right).
 \label{Robertson-Walker metric}
 \end{align}
 Then $\phi$ is assumed to be a real factor such that
 \begin{align*}
 \phi=\,\exp\varphi\left(\frac{\dot{r}^2}{\dot{t}^2}\right)-1> 0.
 \end{align*}
 It is direct that $\phi$ defined in this way is $0$-homogeneous and constant on $M$. It can be re-written covariantly if we consider the vector fields
 \begin{align*}
 \mathcal{A}=\,\frac{1-\epsilon r^2}{a^2(t)}\,\frac{\partial}{\partial r},\quad \mathcal{B}=\,\frac{\partial}{\partial t},
 \end{align*}
 in which case
 \begin{align*}
 \phi(\dot{x})=\,\exp\varphi\left(\frac{h^2(\dot{x},\mathcal{A})}{h^2(\dot{x},\mathcal{B})}\right)-1,
 \end{align*}
 where $h$ is the Robertson-Walker metric \eqref{Robertson-Walker metric}.  Here the vector field $\mathcal{B}$ is the inherent time orientation of the model, which is globally defined. The vector field $\mathcal{A}$ appears as a time dependent force field. In the case when $\epsilon>0$, the force is decreasing with $r$, achieving the maximal intensity at $r=0$ and the minimal intensity equal to zero in the boundary $r=\epsilon^{-1/2}$. Note that in this example $\mathcal{A}$ is spacelike, instead than timelike. However, this does not affect the issue of the singularities, which are found along spacelike vectors only.

 The {\it deformed Robertson-Walker metric} is then
 \begin{align}
g(x,dot{x})=\,\exp\varphi\left(\frac{h^2(\dot{x},\mathcal{A})}{h^2(\dot{x},\mathcal{B})}\right)
 \left(-dt^2+\,a^2(t)\left(\frac{d{r}^2}{1-\epsilon \,r^2}+r^2(\sin^2\theta d{\phi}^2+d{\theta}^2)\right)\right).
 \label{deformed Robertson-Walker spacetime}
 \end{align}
 Note that the denominator in  the argument of $\varphi$ is non-zero for each timelike vector $\dot{x}\in T_x M$. It is direct that the condition \eqref{differential constraint on phi} holds good, despite the fact that the vector field $\mathcal{A}$ is not constant on the cosmic time $t$. The red-shift associated to $g$ coincides with the red shift of $h$. However, the analogous property holds for the associated  Hubble law and deceleration parameter, which  inherit a non-trivial anisotropic perturbation due to the anisotropy of the fundamental tensor.
 \end{ejemplo}
\subsection*{The Einstein field equations for a generalized Berwald spacetime of the type \eqref{generalBerwaldtype}} An important advantage of
generalized Berwald spacetimes respect to generic generalized Finsler spacetimes is that the only non zero curvature endomorphisms of the Berwald connection $^b\nabla$ are the Riemannian type endomorphisms. This will imply that the field equations for a generalized Berwald spacetime can be constructed using the Riemannian curvatures and in close analogy with the field equations of general relativity. Let us discuss such equations in what follows.
\begin{proposicion}
For the metric \eqref{generalBerwaldtype} and for any three vector fields $X,Y,Z\in \,\Gamma \tilde{N}$ the relation
\begin{align*}
R_g(X,Y)Z=\,R_h(X,Y)Z
\end{align*}
between the curvature endomorphisms of $g$ and $h$ holds good.
\label{relationbetweenendomorphisms}
 \end{proposicion}
 \begin{proof}
Since the Levi-Civita connection of $h$ coincides with the linear Berwald connection of $g$, after an identification of the fibers and bundles where they operate, the corresponding curvature endomorphisms are identical in the same sense.
 \end{proof}
 The Ricci tensor of $g$ is defined in close analogy with the Lorentzian case, as the trace of the Riemann curvature endormorphism operator $R_g(\cdot,\cdot)$. Then it is easy to show that the Ricci tensor of $g$ and $h$ coincide,
 \begin{proposicion}
For the metric \eqref{generalBerwaldtype} and given two arbitrary vector fields $X,Y\in\,\Gamma \tilde{N}$, the relation
\begin{align*}
Tr(R_g(X,Y))=\,Tr(R_h(X,Y))
\end{align*}
between the Ricci tensors of $g$ and $h$ holds good.
\label{relationbetweenRiccitensors}
 \end{proposicion}
 As a consequence of {\it Proposition} \ref{relationbetweenendomorphisms} and {\it Proposition} \ref{relationbetweenRiccitensors} the Einstein tensor of $g$ and the Einstein tensor of $h$ coincide,
 \begin{theorem-C}
 For the metric \eqref{generalBerwaldtype} it holds that
 \begin{align}
 ^gR_{\mu\nu}-\,\frac{1}{2}\,R_g\,g_{\mu\nu}=\,^hR_{\mu\nu}-\,\frac{1}{2}\,R_h\,h_{\mu\nu}.
 \label{relationbetweenEinsteintensors}
 \end{align}
 \label{Theorem C}
 \end{theorem-C}
 \begin{proof}
By {\it Proposition} \ref{relationbetweenRiccitensors}, it holds that $^gR_{\mu\nu}=\,^hR_{\mu\nu}$. Moreover,
 \begin{align*}
 R_g=g^{\mu\nu}\,\,^gR_{\mu\nu}=\,(1+\phi)^{-1}\,h^{\mu\nu}\,\,^hR_{\mu\nu}=\,(1+\phi)^{-1}\,R_h,
 \end{align*}
 that directly leads to the result.
 \end{proof}
 Therefore, if $h$ is a solution of the Einstein's vacuum equations, then $g$ is also a solution of the corresponding equations.
 This implies that from the point of view of the description of the gravitational interaction, the Berwald metric determined by $g$ and $h$ are equivalent in vacuum. It also shows that the corresponding Einstein equations for $g$ and for $h$ in presence of matter must be equivalent, because necessarily the left hand side of the Einstein equation for $g$ is the same than the left hand side of the Einstein equation for $h$. Therefore, the right hand sides of such equations must be equivalent. Then we can write the field equations of $g$ to be a direct generalization of Einstein's equations,
 \begin{align}
  ^gR_{\mu\nu}-\,\frac{1}{2}\,R_g\,g_{\mu\nu}=\,8\,\pi\,G\,T_{\mu\nu}.
  \label{field equations}
 \end{align}
 The stress-energy tensor $T_{\mu\nu}$ is determined  by matter fields $\psi_A:M\to \mathcal{E}$, the metric $g$ and by the Berwald connection of $g$.  However, it must be defined over the manifold $M$. Furthermore, it follows
 that the equation \eqref{field equations} is covariantly consistent,
 \begin{align*}
 ^h\nabla\,\left(\,^gR_{\mu\nu}-\,\frac{1}{2}\,R_g\,g_{\mu\nu}\right)=\,^h\nabla\,\left(\,^hR_{\mu\nu}-\,\frac{1}{2}\,R_h\,h_{\mu\nu}\right) =\,^h\nabla\,T_{\mu\nu}=0.
 \end{align*}
 Moreover, for a generalized Berwald spacetime and after the identification of $^b\nabla$ with the associated affine connection, one has that $^h\nabla$ is equivalent to $^g\nabla$. Then we can write
 \begin{align}
 ^g\nabla\,\left(\,^gR_{\mu\nu}-\,\frac{1}{2}\,R_g\,g_{\mu\nu}\right)=\,^g\nabla\,T_{\mu\nu}=0
 \end{align}

 The above argument experiments a change if a cosmological constant term is allowed in the field equations. In this case, we have the relation
 \begin{align}
^hR_{\mu\nu}-\,\frac{1}{2}\,R_h\,h_{\mu\nu}\,+\Lambda\,h_{\mu\nu}=\, ^gR_{\mu\nu}-\,\frac{1}{2}\,R_g\,g_{\mu\nu}\,+(1+\phi)^{-1}\,\Lambda\,g_{\mu\nu}.
 \label{relationbetweenEinsteintensors with cosmological constant}
 \end{align}
 The theory appears to be non-consistent, contrary  to the  case $\Lambda=0$. A possible resolution of this problem is considered in the following lines. We start with the remark that by consistency with experience, the condition $\phi(\dot{x})\ll 1$ must hold. If in addition the cosmological constant $\Lambda$ is assumed to be small respect to an intrinsic radius of the model, for instance, the inverse of the Ricci scalar, then in the expression
 \begin{align*}
 (1+\phi)^{-1}\,\Lambda =\,\Lambda\,-\phi\,\Lambda\,+\frac{\phi^2}{2}\,\Lambda-...
 \end{align*}
 the leading order term is the first one and coincides with a cosmological term in the field equations. Therefore, if the cosmological constant is non-zero but it is a very small parameter compared with other physical contributions to the structure of the spacetime, then the Einstein equations with cosmological constant for $g$ are consistent at leading order in the product $\phi^a\,\Lambda,\,a=0,1,2,...$.

\begin{ejemplo}
If we consider the example given by the spacetime \eqref{deformed Robertson-Walker spacetime}, one has the expression
 \begin{align*}
 (1+\phi)^{-1}\,\Lambda\sim\,\Lambda -\,\varphi\left(\frac{\dot{r}^2}{\dot{t}^2}\right)\,\Lambda+...,
 \end{align*}
where the higher order contributions in $\varphi$ are small compared with $1$ in the open cone $\dot{r}^2/{\dot{t}^2}<1$. A natural way to settle a small value for $\varphi$ is to introduce another length scale, apart from the curvature $\epsilon\neq 0$ inherent from the underlying Robertson-Walker model. The additional scale $\l$ could be associated with a microscopic scale and $\varphi$ can be considered to be an analytical function of the form $\varphi(\l^2\epsilon\,\dot{r}^2/{\dot{t}^2})$. In this case, the series development is of the form
 \begin{align}
 \left(1+\varphi\left(\l^2\epsilon\,\frac{\dot{r}^2}{\dot{t}^2}\right)\right)^{-1}\,\Lambda\sim\,\Lambda -\,\varphi_1\,\frac{\dot{r}^2}{\dot{t}^2}\,\l^2\epsilon\,\,\Lambda+...,
 \label{development varphi}
 \end{align}
 where $\varphi_1$ is a constant with dimensions of speed to the inverse square. However, note that this the simplest ansatz and that others are theoretically allowed. Assuming the  development \eqref{development varphi}, we can re-cast $\varphi_1\sim c^{-2}=1$, where $c$ is the speed of light in vacuum. Moreover, if $\l$ is of order of the Planck length, the product $\l^2\,\epsilon \sim 10^{-124}$, which makes higher contributions on the series \eqref{development varphi} small, if the constant coefficients $\varphi_1$ is of order $1$ and $\Lambda$ small.

 Let us remark that the relevant condition in the above argument is $\l^2\,\epsilon\ll 1$. This condition can be interpreted as the existence of two very different scales that correspond to the  minimal and maximal scales where the continuous generalized Berwald spacetime is applicable as an effective model.
 \end{ejemplo}

 \section{Final remarks}
 From the whole category of generalized Finsler spacetimes, the general physical requirement of existence of smooth  free falling local coordinate systems selects a specific type of structures, namely, generalized Berwald spacetimes. We have considered first a specific family of flat generalized Berwald spacetimes. Our models provide examples that can describe massive and massless world lines and in this sense they can be considered complete. They have certain similarities with the model investigated by Miron-Tavakol metric \cite{Bucataru Miron, Miron Tavakol 1994}, but our models differ from Miron-Tavakol metrics in two fundamental aspects. The first is that our metric \eqref{generalBerwaldtype} is homogeneous in the velocity coordinates. The second is that it is of Berwald type and therefore, the Berwald connection is affine. Since the Miron-Tavakol metric satisfies the Ehlers-Pirani-Schild {\it axioms} \cite{Ehlers Pirani Schild}, the same is to be expected for our family of metrics, since by the properties mentioned above one can say that \eqref{generalBerwaldtype} is {\it more Lorentzian} than the Miron-Tavakol metric.

 It was shown that for the metrics considered in {\bf Theorem B}, the Einstein field equations for $g$ follow directly from the Einstein field equations for $h$. This result holds in both situations, in vacuum and in presence of matter fields in the case when $h$ is a solution of the Einstein field equations. The result is based on the fact that for the family of metrics considered, the Einstein tensor of $g$ coincides with the Einstein tensor of $h$ and as a consequence, the matter stress-energy tensor must live on $M$, even if $g$ has a non-trivial dependence on the velocity coordinates. Thus although the stress-energy tensor must be constructed using the $d$-tensor $g$ and the associated covariant derivative $^b\nabla$, it must be defined also on the spacetime manifold $M$.

 The extension of this result to the situation when the cosmological constant is non-zero is not immediate. Consistency on the zero covariant divergence of the right hand side implies that the left hand side of the equations of motion must also have zero covariant derivative. This is a non-trivial requirement, that can be satisfied if one sees the constancy of $\Lambda$ as an approximation to the more general situation where the right hand side and the cosmological term itself are non-local. This solution is, however, not strictly necessary, since one can argued that the model is consistent if the term of the cosmological constant is small compared with the other geometric terms arising in the field equations. An specific example, based on the deformed Robertson-Walker spacetime, suggests that  in the models investigated in this paper as effective description of spacetimes the cosmological constant term should be small, in order to be consistent with the zero covariant condition of the stress energy tensor.

The above arguments suggest  that $g$ and $h$ are physically equivalent. However,  the Cartan tensors  of $g$ and $h$ are different. This fact has consequences for the chronometric properties based on $g$ and $h$.  Such effects should be detectable by  measuring anisotropic effects in the measure of proper time. Obviously, this cast strong constraints on the level of local anisotropy described by our models.

 Several important questions remain open in the theory sketched in this paper. Probably the most urgent one is to understand the nature of the factor $\phi$. If the field equations for $g$ and $h$ are formally the same, then the only data that can constrain $\phi$ are the boundary conditions for $g$. In particular, generalized Berwald spacetimes on the product manifold   of the type \eqref{generalBerwaldtype} can be defined on product manifolds $\bar{M}_3\times \mathbb{R}$. The  asymptotically condition on $\bar{M}_3$ induced from the Berwald structure on $M$  is analogous to the {\it asymptotic Euclidean Riemannian manifold} and the generalized Finsler metric on $\bar{M}_3$ is, in suitable coordinates, of the form
\begin{align*}
g_{ij}(x,\dot{x})=\,\left(1+\phi(\dot{x^1},\dot{x^2},\dot{x^3})\right)\,\delta_{ij}+\mathcal{O}(r^{-1}),\quad i,j=1,2,3.
\end{align*}
 This argument suggests that the physical significance of the scalar factor $\phi$ can be reduced to the investigation of  flat generalized Berwald spacetime case.

Additional assumptions or hypotheses on the nature of the microscopic structure of the spacetime must be  adopted. This has been made explicit in our example of deformed Robertson-Walker spacetime \eqref{deformed Robertson-Walker spacetime} together with our discussion of the cosmological constant term, where we suggested that one can link the cosmological term appearing in the field equations with the factor $\varphi$, if a microscopic  scale $\l$ is introduced. This new scale may be related with the microscopic limits of applicability of the flat Berwald spacetime model. In resume, in this scenario, three scales appear: $\l$, $\epsilon$ and $\Lambda$. The suggestion that $\Lambda$ could be related with $\l$ and $\epsilon$ reduces the number of scales to two (just $\l$ and $\epsilon$). Finally, let us remark that the nature of $\l$ should be specified from the specific microscopic theory of spacetime.

\subsection*{Acknowledgements.} We thank A. Fuster for several comments on previous versions of this paper. We also thanks V. Perlick for the suggestion leading to explore the Einstein tensor of $g$ and for several other relevant comments. This work was supported by PNPD-CAPES n. 2265/2011, Brazil.

\footnotesize{
}

\end{document}